\documentclass[10pt,conference]{IEEEtran}%
\usepackage{latexsym,delarray,multicol}
\usepackage{epsfig}
\usepackage{amsfonts}
\usepackage{mathrsfs}
\usepackage{amsthm}
\usepackage{amsmath}
\usepackage{amssymb}
\usepackage{revsymb}
\usepackage{graphicx}
\usepackage{url}
\usepackage{cite}%

\newtheorem{proposition}{Proposition}

\newtheorem{lemma}{Lemma}
\begin{document}

\title{Numerical Techniques for Finding the Distances of Quantum Codes}
\date{\today}
\author{\IEEEauthorblockN{Ilya
Dumer\IEEEauthorrefmark{1}, Alexey A. Kovalev\IEEEauthorrefmark{2},
and Leonid P.~Pryadko\IEEEauthorrefmark{3}}
\IEEEauthorblockA{\IEEEauthorrefmark{1} Department of Electrical
Engineering, University of California, Riverside, USA
(e-mail: dumer@ee.ucr.edu)}
\IEEEauthorblockA{\IEEEauthorrefmark{2} Department of Physics,
University of Nebraska at Lincoln, USA (e-mail: alexey.kovalev@unl.edu)}
\IEEEauthorblockA{\IEEEauthorrefmark{3} Department of
Physics \& Astronomy, University of California, Riverside, USA
(e-mail: leonid@landau.ucr.edu)}}
\maketitle

\begin{abstract}
We survey the existing techniques for calculating code distances of  classical
codes and apply these techniques to generic quantum  codes. For classical and
quantum LDPC codes, we also present a new  linked-cluster technique. It
reduces complexity exponent of all  existing deterministic techniques designed
for codes with small  relative distances (which include all known families of
quantum LDPC  codes), and also surpasses the probabilistic technique for
sufficiently high code rates.

\end{abstract}

\section{Introduction}

Quantum error correction (QEC)
\cite{shor-error-correct,Knill-Laflamme-1997,Bennett-1996} is a critical part
of quantum computing due to fragility of quantum states. To date, surface
(toric) quantum codes \cite{kitaev-anyons,Dennis-Kitaev-Landahl-Preskill-2002}
and related topological color
codes\cite{Bombin-MartinDelgado-2006,Bombin-2007, Bombin-MartinDelgado-2007}
have emerged as prime contenders
\cite{Raussendorf-Harrington-2007,Katzgraber-Bombin-MartinDelgado-2009} in
efficient quantum design due to two important advantages. Firstly, they only
require simple local gates for quantum syndrome measurements, and secondly,
they efficiently correct errors below a threshold of about 1\% per gate.
Unfortunately, the locality also limits\cite{Bravyi-Poulin-Terhal-2010} such
codes to an asymptotically zero rate $k/n$. This would make a useful quantum
computer prohibitively large. Therefore, there is much interest in designing
of feasible quantum codes with no locality restrictions.

A more general class of codes is quantum low-density-parity-check (LDPC) codes
\cite{Postol-2001,MacKay-Mitchison-McFadden-2004}. These codes assume no
locality but only require that stabilizer generators (parity checks) have low
weight. Unlike surface or color codes, quantum LDPC codes can have a finite
rate $k/n$. Also, long LDPC codes have a nonzero error probability threshold,
both in the standard setting when a syndrome is measured exactly, and in a
fault-tolerant setting, when syndrome measurements include
errors\cite{Kovalev-Pryadko-FT-2013}. This non-zero error threshold is even
more noteworthy given that known quantum LDPC codes have distances $d$ scaling
as a square root of $n$ unlike linear scaling in the classical LDPC
codes\cite{Tillich-Zemor-2009,Kovalev-Pryadko-2012,
Kovalev-Pryadko-Hyperbicycle-2013,Andriyanova-Maurice-Tillich-2012}. LDPC
codes can have finite rate and linear distance\cite{Bravyi-Hastings-2013} if
weights of stabilizer generators scale as a square root of $n$.
An important open problem is to find the bounds on distance $d$ of quantum
LDPC codes with limited-weight stabilizer generators.

This paper addresses numerical algorithms for finding distances of quantum and
classical LDPC codes. To make a valid comparison, we first survey several
existing classical algorithms that were used before for generic random codes
meeting the Gilbert-Varshamov (GV) bound. Here we re-apply these techniques to
find the distances of quantum codes. Then we turn to the new techniques that
are specific for LDPC codes. Note that most error patterns for such codes form
small clusters that affect disjoint sets of stabilizer
generators~\cite{Kovalev-Pryadko-FT-2013}. While some errors can have huge
weight, they can be always detected if the size of each cluster is below the
code distance $d$. We then design an algorithm that verifies code distance by
checking the error patterns that correspond to the connected error clusters.
For any error weight $w\ll n$, such clusters form an exponentially small
fraction of generic errors of the same weight. Therefore, we consider the
worst-case scenario that holds for any LDPC code and can be applied in quantum
setting. This cluster-based algorithm exponentially reduces the complexity of
the known deterministic techniques for sufficiently small relative distance,
which is the case for all known families of weight-limited quantum LDPC codes.
The new algorithm also outperforms probabilistic techniques for high-rate
codes with small relative distance.

\section{Background}

Let $\mathcal{C}[n,k]_{q}$ be a linear $q$-ary code of length $n$ and
dimension $k$ in the vector space $\mathbb{F}_{q}^{n}$ over the field
$\mathbb{F}_{q}$. This code is uniquely specified by the parity check matrix
$H$, namely $\mathcal{C}=\{\mathbf{c}\in\mathbb{F}_{q}^{n}|H\mathbf{c}=0\}$.
Let $d$ denote the Hamming distance of code $\mathcal{C}$.

A quantum $[[n,k]]$ (qubit) stabilizer code $\mathcal{Q}$ is a $2^{k}%
$-dimensional subspace of the $n$-qubit Hilbert space $\mathbb{H}_{2}^{\otimes
n}$, a common $+1$ eigenspace of all operators in an Abelian \emph{stabilizer
group} $\mathscr{S}\subset\mathscr{P}_{n}$, $-\openone\not \in \mathscr{S}$,
where the $n$-qubit Pauli group $\mathscr{P}_{n}$ is generated by tensor
products of the $X$ and $Z$ single-qubit Pauli operators. The stabilizer is
typically specified in terms of its generators, $\mathscr{S}=\left\langle
S_{1},\ldots,S_{n-k}\right\rangle $; measuring the generators $S_{i}$ produces
the \emph{syndrome} vector. The weight of a Pauli operator is the number of
qubits it affects. The distance $d$ of a quantum code is the minimum weight of
an operator $U$ which commutes with all operators from the stabilizer
$\mathscr{S}$, but is not a part of the stabilizer, $U\not \in \mathscr{S}$.

A Pauli operator $U\equiv i^{m}X^{\mathbf{v}}Z^{\mathbf{u}}$, where
$\mathbf{v},\mathbf{u}\in\{0,1\}^{\otimes n}$ and $X^{\mathbf{v}}=X_{1}%
^{v_{1}}X_{2}^{v_{2}}\ldots X_{n}^{v_{n}}$, $Z^{\mathbf{u}}=Z_{1}^{u_{1}}%
Z_{2}^{u_{2}}\ldots Z_{n}^{u_{n}}$, can be mapped, up to a phase, to a
quaternary vector, $\mathbf{e}\equiv\mathbf{u}+\omega\mathbf{v}$, where
$\omega^{2}\equiv\overline{\omega}\equiv\omega+1$. A product of two quantum
operators corresponds to a sum ($\bmod\,2$) of the corresponding vectors. Two
Pauli operators commute if and only if the \emph{trace inner product}
$\mathbf{e}_{1}\ast\mathbf{e}_{2}\equiv\mathbf{e}_{1}\cdot\overline
{\mathbf{e}}_{2}+\overline{\mathbf{e}}_{1}\cdot\mathbf{e}_{2}$ of the
corresponding vectors is zero, where $\overline{\mathbf{e}}\equiv
\mathbf{u}+\overline{\omega}\mathbf{v}$. With this map, generators of a
stabilizer group are mapped to the rows of a parity check matrix $H$ of an
\emph{additive} code over $\mathbb{F}_{4}$, with the condition that any two
rows yield a nil trace inner product \cite{Calderbank-1997}. The vectors
generated by rows of $H$ correspond to stabilizer generators that act
trivially on the code; these vectors form the \emph{degeneracy group} and are
omitted from the distance calculation.

An LDPC code, quantum or classical, is a code with a sparse parity check
matrix. For a \emph{regular} $(j,\ell)$ LDPC code, every column and every row
of $H$ have weights $j$ and $\ell$ respectively, while for a $(j,\ell
)$-limited LDPC code these weights are limited from above by $j$ and $\ell$. A
huge advantage of classical LDPC codes is that they can be decoded in linear
time using belief propagation (BP) and the related iterative
methods\cite{Gallager-1962,MacKay-book-2003}. Unfortunately, this is not
necessarily the case for quantum LDPC codes, which have many short loops of
length $4$ in their Tanner graphs. In turn, these loops cause a drastic
deterioration in the convergence of the BP algorithm\cite{Poulin-Chung-2008}.
This problem can be circumvented with specially designed quantum
codes\cite{Kasai-Hagiwara-Imai-Sakaniwa-2012,
Andriyanova-Maurice-Tillich-2012}, but a general solution is not known. One
alternative that has polynomial complexity in $n$ and approaches linear
complexity for very low error rates is the cluster-based decoding
of~\cite{Kovalev-Pryadko-FT-2013}.

\section{Generic techniques for distance calculation}

\label{sec:generic} The problem of verifying the distance of a linear code
(finding a minimum-weight codeword) is related to the decoding problem: find
an error of minimum weight that gives the same syndrome as the received
codeword. The number of required operations $N$ usually scales as an exponent
$N\propto q^{Fn}$ in blocklength $n$, and we characterize the complexity by
the exponent $F=\overline{\lim}$ ($\log_{q}N)/n$ as $n\rightarrow\infty$. For
example, for a linear $q$-ary code with $k$ information qubits, inspection of
all $q^{k}$ distinct codewords has (time) complexity exponent $F=R$, where
$R=k/n$ is the code rate. Given substantially large memory, one can instead
consider the syndrome table that stores the list of all $q^{n-k}$ syndromes
and coset leaders. This setting gives (space) complexity $F=1-R$.

\subsection{Sliding window (SW) technique}

\label{sec:sliding} This technique has been proposed in
Ref.~\cite{Evseev-1983} for correction of binary errors and generalized in
Ref.~\cite{Dumer-1996} for soft-decision decoding (where more reliable
positions have higher error costs). A related technique has also been
considered in Refs.~\cite{Zimmermann-1996,Grassl-2006}. The following
proposition addresses this technique for quantum codes. Let $H_{q}%
(x)=x\log_{q}(q-1)-x\log_{q}x-(1-x)\log_{q}(1-x)$ be the $q$-ary entropy
function. Below we consider both generic stabilizer codes and those that meet
the quantum GV bound
\begin{equation}
R=1-2H_{4}(\delta) \label{eq:GV}%
\end{equation}

\begin{proposition}
Code distance $\delta n$ of a random quantum stabilizer code  $[[n,Rn]]$ can
be found with complexity exponent
\begin{equation}
F_{Aq}=(1+R)H_{4}(\delta) \label{1}%
\end{equation}
For random stabilizer codes that meet the GV bound (\ref{eq:GV}),
\begin{equation}
F_{Aq}^{\ast}=(1-R^{2})/2 \label{eq:SW-GV-complexity}%
\end{equation}

\end{proposition}

\begin{proof}
SW technique uses only $k+o(n)$ consecutive positions to recover a  codeword
of a $q$-ary linear $[n,k]$ code. For example, any $k$  consecutive positions
suffice in a cyclic code. It is also easy to  verify that in most (random)
$k\times n$ generator matrices $G$ any  $s=k+2\left\lfloor \log_{q}
n\right\rfloor $ consecutive columns form a submatrix $G_{s}$ of a  maximum
rank $k$. Thus, in most random $[n,k]$ codes, a codeword can  be recovered by
encoding its $s$ (error free) consecutive bits. To  find a codeword $c$ of any
given weight $w$, we choose a sliding  window $I(i,s)$ of length $s$ that
begins in position  $i=0,\ldots,n-1$. Note that a sliding window can change
its weight  only by one when it moves from any position $i$ to $i+1$; thus at
least one of the $n$ windows will have the average Hamming weight
$v\equiv\left\lfloor ds/n\right\rfloor $. Our algorithm takes all  possible
positions $i$ and weights $w=1,2,\ldots$. We assume that  the current window
$I(i,s)$ is corrupted in $v$ positions and encode  all
\begin{equation}
L=(q-1)^{v}\textstyle{\binom{s}{v}} \label{eq:sliding-wind}%
\end{equation}
vectors of length $s$ and weight $v$. Procedure stops for some $w$ once we
find an encoded codeword $c$ of weight $w$. Finally, such vector $c$ is tested
on linear dependence with the rows of the parity check matrix $H$. This gives
the overall SW-complexity of order $Ln^{2}$ with complexity exponent
$F_{A}=RH_{q}(\delta)$.

To apply SW procedure to a (degenerate) quantum code, note that an $[[n,k]]$
stabilizer code is related to some additive quaternary code that is defined in
a space of $4^{n}$ vectors and has only $2^{n-k}=4^{r/2}$ distinct syndromes,
where $r\equiv n-k$ is the redundancy of the quantum code. Thus, the effective
rate is\footnote{This construction is analogous to pseudogenerators
introduced in Ref.~\cite{White-Grassl-2006}.} $R^{\prime}=(n-r/2)/n=(1+R)/2$,
which gives binary complexity exponent (\ref{1}). Finally,
estimate~(\ref{eq:SW-GV-complexity}) follows from~(\ref{eq:GV}).
\end{proof}

Note that classical codes that meet the GV bound $R=1-H_{q}(\delta)$ have
complexity exponent $F_{A}^{\ast}=R(1-R)$ that achieves its maximum $1/4$ at
$R=1/2$. By contrast, quantum codes achieve maximum complexity $F_{Aq}^{\ast
}(R)$ at the rate $R=0$. Note also that quantum codes of low rate $R$ and
small relative distance $\delta$ have complexity exponent logarithmic in
$\delta$.

\begin{figure}[ptbh]
\centering
\includegraphics[width=0.8\columnwidth]{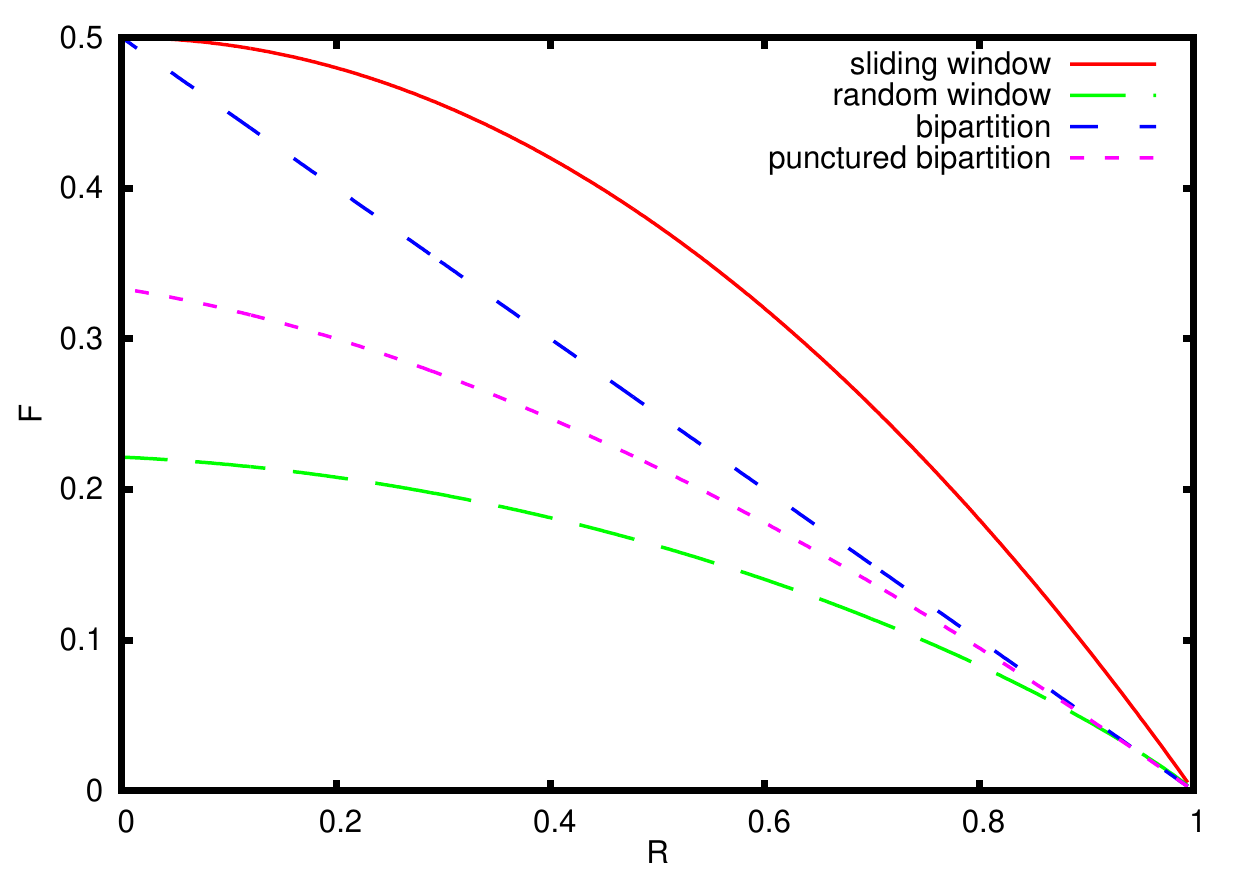}\caption{Comparison of the binary
complexity exponents for the four classical decoding techniques applied to
quantum codes at the quantum GV bound, see Sec.~\ref{sec:generic}. Note that
for high-rate codes, $R\to1$, the curves for the sliding window and the random
window techniques have logarithmically-divergent slopes, while the slopes for
the two other techniques remain finite. In this limit of $R\to1$ the punctured
bipartition technique gives the best performance.}%
\label{fig:cmp}%
\end{figure}

\subsection{Random window (RW) technique\cite{Leon-1988,
Kruk-1989,Coffey-Goodman-1990}}

\begin{proposition}
Code distance $\delta n$ of a random quantum stabilizer code $[[n,Rn]]$ can be
found with complexity exponent%
\begin{equation}
F_{Bq}=H_{2}(\delta)-\textstyle\left(  \frac{1-R}{2}\right)  H_{2}\left(
\frac{2\delta}{1-R}\right)  \label{3}%
\end{equation}

\end{proposition}

\begin{proof}
Given a random $q$-ary linear $[n,k]$ code, we randomly choose  $s=k+\tau$
positions, where $\tau=o(k)$ is some small positive  number, e.g., $\tau
\sim\log_{2} k$. We wish to find an $s$-set  of weight $t=1$ in some unknown
codeword of weight $w$. Let  $M(n,s,w)$ denote the number of random trials
needed to find such  a set with a high probability $1-e^{-n}$. Also, let
$T(n,s,w)$ be  the minimum number of $(n-s)$-sets needed to necessarily cover
any  (unknown) $w$-set. It is easy to check\cite{Erdos-book} that
\begin{equation}
\textstyle{\binom{n}{w}}/{\binom{n-s}{w}}\leq T(n,s,w)\leq\textstyle{\binom
{n}{w}}/{\binom{n-s}{w}}(1+\ln{\binom{n-s}{w}}) \label{eq:112}%
\end{equation}
and that $M(n,s,w)\leq T(n,s,w)n\ln n$. Below $w=1,2,\ldots$.

RW-algorithm performs $nT(n,s,w)$ trials of choosing $s$ random positions.
Each trial gives a random $k\times s$ submatrix $G_{s}$ of a (random)
generator matrix $G$. It is easy to verify that $G_{s}$ has full rank $k$ with
a high probability $1-q^{-\tau}$ (also, most matrices $G$ have \textit{all
possible } submatrices $G_{k}$ of rank $k-n^{1/2}$ or more$)$. Thus, a typical
$s$-set has a subset of $k$ information bits. If the current $s$-set includes
such a subset, we consider $s$ vectors $(0\ldots010\ldots0)$ of weight $t=1$
and re-encode them into the codewords of length $n$. Otherwise, we discard the
$s$-set and proceed further. Algorithm stops once we obtain a codeword of
weight $w$. The overall complexity has the order of $n^{4}T(n,s,w)$ with the
binary complexity exponent
\[
F_{B}=H_{2}(\delta)-(1-R)H_{2}(\delta/(1-R)).
\]
For stabilizer codes, we obtain (\ref{3}) using their effective rate
$R^{\prime}=(1+R)/2$.
\end{proof}

Quantum codes with small distances $w\leq(n-k)^{1/2}$ and $s\sim nR^{\prime}$
meet an exponentially tight bound
\[
\log_{2}T(n,s,w)\sim\log_{2}\textstyle\left(  \frac{n}{n-s}\right)  ^{w}\sim
w-w\log_{2}(1-R)
\]
Exponent (\ref{3}) can be further specified if codes meet the quantum GV bound
(\ref{eq:GV}). The corresponding exponent $F_{Bq}^{\ast}$ reaches its maximum
$F_{\mathrm{max}}\approx0.22$ at $R=0$ [Fig.~\ref{fig:cmp}]. By contrast,
bi\-na\-ry linear codes give exponent $F_{B}^{\ast}=(1-R)\bigl[ 1-H_{2}\left(
\delta/(1-R)\right)  \bigr] $ that achieves its maximum $0.12$ at
$R\approx1/2$.

\subsection{Bipartition match (BM) technique\label{sec:bipartition}%
\cite{Dumer-1989}}

\begin{proposition}
Code distance $\delta n$ of any quantum stabilizer code $[[n,Rn]]$  can be
found with complexity exponent%
\begin{equation}
F_{Cq}=H_{4}(\delta). \label{4}%
\end{equation}
For random stabilizer codes that meet the GV bound (\ref{eq:GV}),
\begin{equation}
F_{Cq}^{\ast}=(1-R)/2. \label{5}%
\end{equation}

\end{proposition}

\begin{proof}
We use a sliding (\textquotedblleft left\textquotedblright) window  of length
$s_{l}=\left\lfloor n/2\right\rfloor $ starting in any  position $i$. For any
unknown vector of weight $w$, at least one  position $i$ produces a window of
the average weight (down to the  closest integer) $v_{l}=\left\lfloor
w/2\right\rfloor $. The  remaining (right) window of length $s_{r}=\left\lceil
n/2\right\rceil $ will have the weight $v_{r}=\left\lceil w/2\right\rceil $.
We calculate the syndromes of all vectors  $e_{l}$ and $e_{r}$ of weights
$v_{l}$ and $v_{r}$ on the left and  right windows, respectively, and try to
find two vectors  $\{e_{l},e_{r}\}$ that give identical syndromes, and
therefore form  a codeword. Clearly, each set $\{e_{l}\}$ and $\ \{e_{r}\}$
have  size of order $L=(q-1)^{w/2}{\binom{n/2}{w/2}}$. Finding two  elements
$e_{l}$, $e_{r}$ with equal syndromes requires complexity  of order $L\log
_{2}L$, by sorting the elements of the combined  set. Thus, finding a code
vector of weight $w=\delta n$ in any  classical code requires complexity of
order $q^{F_{C}n}$, where  $F_{C}=H_{q}(\delta)/2$. For binary codes on the GV
bound,  $F_{C}^{\ast}=(1-R)/2$. The arguments of the previous propositions
then  give exponents (\ref{4}) and (\ref{5}) for stabilizer codes.
\end{proof}

Note that BM-technique works for any linear code, unlike two previous
techniques provably valid for random codes. It is also the only technique that
can be transferred to quantum codes without any performance loss. Note also
that $F_{Cq}^{\ast}$ is always below the SW exponent $F_{Aq}^{\ast}$, and is
below the RW exponent $F_{Bq}^{\ast}$ for very high rates. This is due to the
fact that $F_{Bq}^{\ast}\approx1-R$ for $R\rightarrow1$, and is twice the
value of $F_{Cq}^{\ast}$.

\subsection{Punctured bipartition technique \cite{Dumer-2001}}

\begin{proposition}
Code distance $\delta n$ of a random quantum stabilizer code $[[n,Rn]]$ can be
found with complexity exponent%
\begin{equation}
F_{Dq}=\textstyle{\frac{2(1+R)}{3+R}}H_{4}(\delta)
\label{eq:punctured-complexity}%
\end{equation}
For random stabilizer codes that meet the GV bound (\ref{eq:GV}),
\begin{equation}
F_{Dq}^{\ast}=(1-R^{2})/(3+R) \label{eq:punctured-GV}%
\end{equation}

\end{proposition}

\begin{proof}
We now combine the SW and BM techniques and consider a sliding  window of
length $s=\left\lceil 2nR/(1+R)\right\rceil $ that exceeds  $k$ by a factor of
$2/(1+R)$. It is easy to verify that most random  $[n,k]$ codes include at
least one information $k$-set on any  sliding $s$-window $I(i,s)$. Thus, any
such window forms a punctured  linear $[s,k]$ code with a smaller redundancy
$s-k$. Also, any  codeword of weight $w$ has the average weight
$v=\left\lfloor ws/n\right\rfloor $ in one or more sliding windows. For
simplicity, let $s$ and $v$ be even. We then use bipartition on each
$s$-window and consider all vectors $e_{l}$ and $e_{r}$ of weight  $v/2$ on
each half of length $s/2$. The corresponding sets  $\{e_{l}\}$ and $\{e_{r}\}$
have size $L_{s}=(q-1)^{v/2} \left(  _{v/2}^{s/2}\right)  $. We then seek all
matching pairs  $\{e_{l} ,e_{r}\}$ that have the same syndrome $h$. Each such
pair  $\{e_{l},e_{r}\}$ represents some code vector of the punctured  $[s,k]$
code and is re-encoded to the full length $n$. For each  $w=1,2,\ldots$, we
stop the procedure once we find a re-encoded vector  of weight $w$. Thus, we
use $[s,k]$ punctured codes and lower  BM-complexity to the order $L_{s}$.

However, it can be verified that some (short) syndrome $h$ of size  $s-k$ can
appear in many vectors $e_{l}$ and $e_{r}$ of length  $s/2$, unlike the
original BM-case. It turns out \cite{Dumer-2001}  that our choice of parameter
$s$ limits the number of such  combinations $e_{l}, e_{r}$ to the above order
$L_{s}$. Thus, we  have to encode \textit{all }$L_{s}$ \textit{code vectors}
of weight  $v$ in a random $[s,k]$ code. The end result is a smaller
complexity  of order $L_{s}=q^{F_{D}n}$, where
\[
F_{D}=H_{q}(\delta)R/(1+R).
\]
Transition from classical codes to quantum codes does not affect
BM-complexity. However, our sliding algorithm again depends on the effective
quantum code rate $R^{\prime}=(1+R)/2$. This changes exponent $F_{D}$ for
classical codes to exponent (\ref{eq:punctured-complexity}) for stabilizer
codes. Quantum GV bound (\ref{eq:GV}) gives (\ref{eq:punctured-GV}).
\end{proof}

For random codes of high rate $R\rightarrow1$ that meet the GV bound, this
technique gives the lowest known exponents $F_{Dq}^{\ast}$ (for stabilizer
codes) and $F_{D}^{\ast}=R(1-R)/(1+R)$ (for binary codes). However, it cannot
be provably applied to any linear code, unlike a simpler bipartition
technique. Finally, the above propositions can be applied to a narrower class
of the Calderbank-Shor-Steane (CSS) codes. Here a parity check matrix is a
direct sum $H=G_{x}\oplus\omega G_{z}$, and the commutativity condition
simplifies to $G_{x}G_{z}^{T}=0$. A CSS code with $\mathrm{rank}%
\mathop G_{x}=\mathrm{rank}\mathop G_{z}=(n-k)/2$ has the same effective rate
$R^{\prime}=(1+R)/2$ since both codes include $k^{\prime}=n-(n-k)/2=(n+k)/2$
information bits. It is readily verified that CSS codes have binary complexity
exponents $F(R,\delta)$ given by expressions similar to Eqs.\ (\ref{1}),
(\ref{3}), (\ref{4}), (\ref{eq:punctured-complexity}), where one must
substitute $H_{4}(x)$ with $H_{2}(x)/2$.

\section{Linked-cluster technique}

Let $\Psi(s,\ell)$ be an ensemble of \emph{regular} $(s,\ell)$ LDPC codes, in
which every column and every row of matrix $H$ has weight $s$ and $\ell$
respectively. The following technique is designed as an alternative to the BP
technique used in \cite{Hu-Fossorier-Eleftheriou-2004} to find code distance.
First, note that with quantum codes, BP can yield decoding
failures\cite{Poulin-Chung-2008}, while our setting requires error-free
guarantee. The second, more important, reason is that we consider very
specific, self-orthogonal LDPC codes that can be used in quantum setting.
These self-orthogonal codes represent very atypical elements of $\Psi(s,\ell)$
and can have drastically different parameters. In particular, the existing
constructions of such codes have low distance $d$, where $\log d\sim(\log
n)/2$, whereas a typical $(s,\ell)$-code has a linearly growing distance.
Thus, we consider the worst-case scenario in $\Psi(s,\ell)$, which can be
provably applied to any code.

For an $(s,\ell)$-code, we represent all (qu)bits as nodes of a graph
$\mathcal{G}$ with vertex set $V(\mathcal{G})$ and connect two nodes iff there
is a parity check that includes both positions. A codeword $\mathbf{c}$ is
defined by its support $\mathcal{E}\subseteq V(\mathcal{G})$ and induces the
subgraph $\mathcal{G}(\mathcal{E})$ that forms one or more \textit{clusters}
and has no edges outside of $\mathcal{G}(\mathcal{E})$. Generally, we will
make no distinction between the set $\mathcal{E}$ and the corresponding
subgraph. Note that disconnected clusters affect disjoint sets of the parity
checks. This implies

\begin{lemma}
[Lemma 1 from Ref.\ \cite{Kovalev-Pryadko-FT-2013}] \label{th:basic} A
minimum-weight code word of a $q$-ary linear code  forms a linked cluster on
$\mathcal{G}$.
\end{lemma}

\begin{proof}
Let a minimum-weight support $\mathcal{E}$ include disconnected  parts, say
$\mathcal{E}_{1}$ and $\mathcal{E}_{2}$. These parts  satisfy different parity
checks. Then vectors generated by  $\mathcal{E}_{1}$ and $\mathcal{E}_{2}$
belong to our code and have  smaller weights. Contradiction.
\end{proof}

\emph{Linked-cluster algorithm. } The following breadth-first
algorithm inspects all fully-linked clusters of a given weight
$w=\delta n$. Let us assume that $j=0,1,\ldots,n-1$ is the
starting position in the support $\mathcal{E}$ of an unknown
codeword of weight $w$. Position $j$ belongs to some $s$ parity
checks which form the list $\eta=\{h_{1},\ldots,h_{s}\}$. To
satisfy the parity-check $h_{1}$, we arbitrarily select some (odd)
number $v_{1}$ of the remaining\ $\ell-1$ parity-check positions
of $h_{1}$. These $v_{1}$ positions are now included in the
current support $\mathcal{E}$. Any time a new position is
selected, we also append  the list $\eta$ with the new parity
checks which include this position. We then proceed with the
subsequent parity-checks $h_{2},h_{3},\ldots$ as follows. Let a
check $h_{i}$  overlap with some of the parity checks
$h_{1},\ldots,h_{i-1}$ in $a_{i}\leq\ell-1$ positions, and let
$b_{i}$ be the number of 1s selected in these $a_{i}$ positions.
Then $h_{i}$ can use only the remaining $\ell-a_{i}$ positions to
pick up some $v_{i}\equiv b_{i}(\bmod\,2)$ positions. If $b_{i}$
is odd, the algorithm adds some $v_{i}\in\{1,3,\ldots\}$ positions
from $h_{i}$, but (temporarily) skips this check if $b_{i}$ is
even. This parity check $h_{i}$ can be re-processed in some later
step $p$ as a parity check $h_{p}$ if the corresponding number
$b_{p}$ is odd. The process is stopped once we add $v=w-1$
positions. The result is a binary codeword with support
$\mathcal{E}$ if all processed odd-type checks are satisfied and
all unprocessed checks have even overlap with $\mathcal{E}$. At
this point, adding some $v_{i}=2,4,..$ symbols in any even-type
check can only increase the weight of a codeword. For a $q$-ary
code, we perform summation over $v_{i}=1,2,\ldots$, and need to
check the rank of a matrix formed by the corresponding $w$ columns
of the check matrix. For a quantum stabilizer code, we also need
to verify that any obtained codeword is linearly independent from
rows of $H$.

At step $i$, there are $\binom{\ell-a_{i}}{v_{i}}$ ways to select
$v_{i}$ positions. Thus, the total number of choices $N_{v}$ to
select $v$ positions is
\[
N_{v}\leq\sum_{m\geq1}\,\,\sum_{v_{i}\in\{1,3,\ldots\}}\delta_{v,v_{1}%
+\ldots+v_{m}}\prod_{i=1}^{m}\textstyle\binom{\ell-a_{i}}{v_{i}}%
\]
which in turn is bounded by
\[
S_{v}(\ell)\equiv\sum_{m\geq1}\,\,\sum_{v_{i}\in\{1,3,\ldots\}}\delta
_{v,v_{1}+\ldots+v_{m}}\prod_{i=1}^{m}\textstyle\binom{\ell-1}{v_{i}}%
\]
Here $\delta_{a,b}$ is the Kronecker symbol, and $m$ is the number of terms in
the decomposition $v=v_{1}+\ldots+v_{m}$.

To estimate $S_{v}(\ell)$, introduce the generating function $g_{\ell}%
(z)=\sum_{v}S_{v}(\ell)z^{v}$. Easy summation gives for $q=2$:
\begin{equation}
g_{\ell}(z)=\textstyle\left(  1-f_{\ell}(z)\right)  ^{-1},\quad f_{\ell
}(z)\equiv\frac{(1+z)^{\ell-1}-(1-z)^{\ell-1}}{2}%
\end{equation}
Finally, we use the contour integration of $g_{\ell}(z)$ to find the
coefficients $S_{v}(\ell)$. Let $\gamma=\sinh^{-1}(1)\approx1.135$ in the case
of a binary code, and $\gamma=1/\ln2\approx1.443$ in the case of a $q$-ary
code. We have:

\begin{proposition}
A codeword of weight $\delta n$ in a $(s,\ell)$-code can be found  with
complexity exponent  $F_{\mathrm{LC}}=\delta\log_{2}(\gamma_{\ell}(\ell-1))$,
where  $\gamma_{\ell} \in(1, \gamma)$ grows monotonically with~$\ell$.
\end{proposition}

More precise estimates of $S_{v}(\ell)$ also give specific numbers
$\gamma_{\ell}$, which can be important for small values of $\ell$.
Finally note that while the cluster technique has high complexity for large
$\ell$ and $\delta$, its exponent $F_{\mathrm{LC}}$ is linear in the relative
distance $\delta$. In comparison, deterministic techniques of
Sec.~\ref{sec:generic} give the higher exponents $F\propto\delta\log
(1/\delta)$ in this limit. All known quantum LDPC codes with limited-weight
stabilizer generators have $\delta\propto n^{-1/2}$, and the linked-cluster
technique gives the lowest complexity for these codes. Note that the RW
technique also gives a linear in $\delta$ exponent $F_{Bq}$ that is bounded by
$\delta-\delta\log_{2}(1-R)$. Our cluster technique still lowers this exponent
$F_{Bq}$ for high code rates $R\geq1-2[\gamma\,(\ell-1)]^{-1}$.

\section{Conclusion}

In this paper, we considered different techniques of finding code distances of
stabilizer quantum codes. For sparse quantum LDPC codes, we proposed a new
cluster-based technique. This technique reduces complexity exponents of the
existing non-probabilistic algorithms for codes with sufficiently small
relative distances. In particular, this is the case for all known families of
quantum LDPC codes that have distances of order $n^{1/2}$ or less.
Cluster-based technique also beats the probabilistic random-window technique
for high-rate codes.\medskip

\textit{Acknowledgment.}
This work was supported in part by the
U.S. Army Research Office under Grant No.\ W911NF-11-1-0027, and by
the NSF under Grant No.\ 1018935.
\bibliographystyle{IEEEtran}
\bibliography{IEEEabrv,lpp,qc_all,more_qc,ldpc,MyBIB}

\end{document}